\documentclass[12pt]{amsart}
\usepackage{amssymb}
\usepackage{amsbsy}
\usepackage{amscd}

\usepackage[mathscr]{eucal}
\usepackage{nicefrac}

\usepackage{a4}
\usepackage[all]{xy}

\usepackage{refcount}

\usepackage{verbatim}
\usepackage{version}
\usepackage{color}
\newenvironment{NB}{
\color{red}{\bf NB}. \footnotesize
}{}

\excludeversion{NB}
\excludeversion{NB2}
\makeatletter

\usepackage{url}
\usepackage{ifmtarg}
\usepackage[
bookmarks=true,
colorlinks=true]{hyperref}
\renewcommand{\thesubsection}{\thesection(\@roman\c@subsection)}
\newenvironment{aenume}{%
  \begin{enumerate}%
  }{\end{enumerate}}
\newcounter{number}
\makeatother
\newtheorem{Lemma}[equation]{Lemma}
\newtheorem{Proposition}[equation]{Proposition}

\theoremstyle{definition}

\theoremstyle{remark}
\newtheorem{Remark}[equation]{Remark}

\newtheorem{Question}[equation]{Question}

\numberwithin{equation}{section}

\newcommand{\propref}[1]{Proposition~\ref{#1}}

\newcommand{\subsecref}[1]{\S\ref{#1}}

\newcommand{\defeq}{\overset{\operatorname{\scriptstyle def.}}{=}}
\newcommand{\CC}{{\mathbb C}}
\newcommand{\ZZ}{{\mathbb Z}}

\newcommand{\RR}{{\mathbb R}}
\newcommand{\proj}{{\mathbb P}}

\newcommand{\SL}{\operatorname{\rm SL}}
\newcommand{\SU}{\operatorname{\rm SU}}
\newcommand{\GL}{\operatorname{GL}}
\newcommand{\PGL}{\operatorname{PGL}}
\newcommand{\U}{\operatorname{\rm U}}

\newcommand{\grpSp}{\operatorname{\rm Sp}}

\newcommand{\gl}{\operatorname{\mathfrak{gl}}}

\newcommand{\g}{{\mathfrak g}}

\newcommand{\Ker}{\operatorname{Ker}}

\newcommand{\Ima}{\operatorname{Im}}

\newcommand{\DB}{\overline{\partial}}

\newcommand{\id}{\operatorname{id}}

\renewcommand{\MR}[1]{}

\newcommand{\dslash}{/\!\!/}
\newcommand{\bM}{\mathbf M}
\newcommand{\bN}{\mathbf N}
\newcommand{\CS}{\mathrm{CS}}

\newcommand{\tslash}{/\!\!/\!\!/}
\newcommand{\tslabar}{\mathbin{
\setbox0=\hbox{/\!\!/\!\!/}\rule[0.4\ht0]{\wd0}{.3\dp0}\kern-\wd0\box0}}
\newcommand{\Hyp}{\operatorname{Hyp}}
\newcommand{\bmu}{\boldsymbol\mu}

\newcommand{\la}{\lambda}
\newcommand{\aff}{\mathrm{aff}}

\newcommand{\cF}{\mathcal F}

\makeatletter
\newcommand{\cA}[1][{}]{%
  \@ifmtarg{#1}%
  {\mathcal A}
  {\mathcal A(#1)}
}
\newcommand{\cAh}[1][{}]{%
  \@ifmtarg{#1}%
  {\mathcal A_\hbar}
  {\mathcal A_\hbar(#1)}
}
\makeatother

\newcommand{\Stab}{\operatorname{Stab}}

\begin{document}

\title[Questions on Coulomb branches of $3d$ $\mathcal N=4$ gauge theories] 
{Questions on provisional 
Coulomb branches of $3$-dimensional $\mathcal N=4$ gauge theories
}
\author[H.~Nakajima]{Hiraku Nakajima}
\address{Research Institute for Mathematical Sciences,
Kyoto University, Kyoto 606-8502,
Japan}
\email{nakajima@kurims.kyoto-u.ac.jp}

\pagenumbering{roman}
\pagestyle{plain}

\subjclass[2000]{}
\begin{abstract}
    This is a supplement to \cite{2015arXiv150303676N}, where an
    approach towards a mathematically rigorous definition of the
    Coulomb branch of a $3$-dimensional $\mathcal N=4$ SUSY gauge
    theory was proposed. We ask questions on their expected
    properties, especially in relation to the corresponding Higgs
    branch, partly motivated by the interpretation of the level rank
    duality in terms of quiver varieties \cite{Na-branching} and the
    symplectic duality \cite{2014arXiv1407.0964B}. We study questions
    in a few examples.
\end{abstract}

\maketitle

\setcounter{tocdepth}{2}

\section*{Introduction}

In \cite{2015arXiv150303676N}, we proposed an approach towards a
mathematically rigorous definition of the Coulomb branch $\mathcal
M_C$ of a $3$-dimensional $\mathcal N=4$ SUSY gauge theory.
Moreover various physically known examples and expected properties
were reviewed.
In this paper, we add questions on expected properties, especically in
relation to the corresponding Higgs branch $\mathcal M_H$. Some are
probably implicit in the physics literature, but we are motivated by
(a) the interpretation of the level rank duality of affine Lie
algebras of type $A$ via quiver varieties, and also (b) the symplectic
duality.

Recall when $\mathcal M_H$ is a quiver variety of affine type $A$, the
corresponding $\mathcal M_C$ is also a quiver variety of affine type
$A$ \cite{MR1454291,MR1454292}. Under their relation to representation
theory \cite{Na-quiver}, the pair is given by I.~Frenkel's level rank
duality \cite{MR675108}. Then one can interpret representation
theoretic statements as relation between $\mathcal M_H$ and $\mathcal
M_C$. This idea was a source of inspiration in author's work
\cite{Na-branching} on provisional double affine Grassmannian
\cite{braverman-2007}, and the joint work \cite{2014arXiv1406.2381B}.

Recall the symplectic duality \cite{2014arXiv1407.0964B} predicts
pairs of symplectic manifolds $(\mathfrak M, \mathfrak M^!)$ whose
categories $\mathcal O$, $\mathcal O^!$ are Koszul dual to each
other. Here the categories $\mathcal O$, $\mathcal O^!$ are certain
full subcategories of modules of quantizations of $\mathfrak M$,
$\mathfrak M^!$.
Many examples of symplectic dual pairs appear as $(\mathfrak M,
\mathfrak M') = (\mathcal M_H, \mathcal M_C)$ for a gauge theory,
e.g., the above quiver varieties of affine type $A$. If $\mathcal O$
and $\mathcal O^!$ are Koszul dual, one can deduce many relations
between $\mathfrak M$ and $\mathfrak M^!$. Therefore it is natural to
ask relations between $\mathcal M_H$ and $\mathcal M_C$. Note also
that the Koszul duality is expected to be closely related to the level
rank duality.

Note that $(\mathcal M_H, \mathcal M_C)$ are much more general: for
example they may not have resolution of singularities nor a torus
action with isolated fixed points, as required in the formulation of
the symplectic duality. They have always $\CC^\times$-action which
scale the symplectic form by weight $2$, but are not cone in general.
Even more fundamentally, we have lots of examples where $\mathcal M_H$
is a point, while $\mathcal M_C$ are nontrivial.
Therefore we decide to restrict ourselves to ask \emph{basic} (or
\emph{naive}) questions which can be asked without introducing
resolutions, torus action.
We study these questions in a few examples.

\subsection*{Notation}

\begin{enumerate}
      \item We basically follow the notation in Part I
    \cite{2015arXiv150303676N}. However we mainly use a complex
    reductive group instead of its maximal compact subgroup. Therefore
    we denote a reductive group by $G$, and its maximal compact by
    $G_c$. We only use the complex part of the hyper-K\"ahler moment
    map, for which we use the notation $\bmu$.

      \item We also change the notation for a compact Riemann surface
    from $C$ to $\Sigma$ to avoid a conflict with `C' for the Coulomb
    branch.
\end{enumerate}

\section{More general target spaces}

This section will be independent of other parts of the paper. The
reader can skip it, but can be also considered as a brief review of
the construction in \cite{2015arXiv150303676N} and its generalization.

In \cite{2015arXiv150303676N}, we define a sheaf of a vanishing cycle
on a moduli space associated with a complex symplectic representation
$\bM$ of a reductive group $G$. We consider their modification and
generalization.

This construction, as well as the previous one in
\cite{2015arXiv150303676N}, should be understood in the framework of
shifted symplectic structures \cite{MR3090262}.\footnote{The author
  thanks Dominic Joyce for pointing out a relevance of
  \cite{MR3090262} in our construction.} Because of the author's lack
of ability, we cannot make it precise unfortunately.

\subsection{$\sigma$-models}

Let us first consider
\begin{itemize}
      \item $(\bM,\omega)$ is a (holomorphic) symplectic manifold with
a $\CC^\times$-action such that $t^*\omega = t^2\omega$ for $t\in\CC^\times$.
  \item there is a Liouville form $\theta$ such that $d\theta = \omega$ and $t^*\theta = t^2\theta$.
\end{itemize}

Let $\Sigma$ be a compact Riemann surface. We choose and fix a square
root $K_\Sigma^{1/2}$ of the canonical bundle $K_\Sigma$. We consider
the associated $\CC^\times$-principal bundle $P_{K_\Sigma^{1/2}}$,
that is $P_{K_\Sigma^{1/2}}\times_{\CC^\times}\CC \cong
K_\Sigma^{1/2}$ when $\CC^\times$ acts on $\CC$ with weight $1$.
\begin{NB}
    Therefore a $\CC^\times$-equivariant map $s\colon
    P_{K_\Sigma^{1/2}}\to \CC$ ($s(p\cdot t) = t^{-1} s(p)$) is a
    section of $P_{K_\Sigma^{1/2}}\times_{\CC^\times}\CC$.
\end{NB}%


We consider a $\CC^\times$-equivariant $C^\infty$-map $\Phi\colon
P_{K_\Sigma^{1/2}}\to \bM$, in other words a $C^\infty$-section of a
bundle $P_{K_\Sigma^{1/2}}\times_{\CC^\times}\bM$.
\begin{NB}
    Let $p\in P_{K_\Sigma^{1/2}}$. Then $\Phi(p)\in\bM$. We then
    define $\Phi^\sim(x) = [p, \Phi(p)]$ for $x\in \Sigma$ such that
    $p$ is in the fiber over $x$. This is well-defined, as $[p\cdot g,
    \Phi(p\cdot g)] = [p\cdot g, g^{-1}\cdot \Phi(p)]
    = [p, \Phi(p)]$. Therefore $\Phi^\sim$ defines a section of
    $P_{K_\Sigma^{1/2}}\times_{\CC^\times}\bM$.
\end{NB}%
Taking a local holomorphic trivialization of $P_{K_\Sigma^{1/2}}$, we
see that the $(0,1)$-part of $\Phi^*\theta$ is a well-defined
$K_\Sigma$-valued $(0,1)$-form, i.e., $(1,1)$-form on $\Sigma$: let
$\{ U_\alpha \}$ be an open cover of $\Sigma$, such that
$P_{K_\Sigma^{1/2}}$ is trivialized over $U_\alpha$. We denote the
transition function by $g_{\alpha\beta}$. Then $\Phi$ is a collection
$\{ \Phi_\alpha\colon U_\alpha\to \bM\}$ such that $\Phi_\alpha =
g_{\alpha\beta}\cdot\Phi_\beta$ on $U_\alpha\cap U_\beta$, where
$\cdot$ is the $\CC^\times$-action on $\bM$. We have $\DB\Phi_\alpha =
g_{\alpha\beta}\cdot \DB\Phi_\beta$ as $g_{\alpha\beta}$ is
holomorphic. Therefore $(\Phi_\alpha^*\theta)^{(0,1)} = \langle\theta,\DB\Phi_\alpha\rangle =
\langle\theta, g_{\alpha\beta}\cdot \DB\Phi_\beta\rangle =
g_{\alpha\beta}^2 \langle\theta, \DB\Phi_\beta\rangle$. This means
that $\{ \langle\theta,\DB\Phi_\alpha\rangle \}$ is a
$K_\Sigma$-valued $(0,1)$-form. This is $(\Phi^*\theta)^{(0,1)}$.

We integrate it over $\Sigma$:
\begin{equation*}
    \CS(\Phi) \defeq \int_\Sigma (\Phi^*\theta)^{(0,1)}.
\end{equation*}
This is the holomorphic Chern-Simons type integral in this setting.
Moreover $\Phi$ is a critical point of $\CS$ if and only if $\Phi$ is
a holomorphic section of $P_{K_\Sigma^{1/2}}\times_{\CC^\times}\bM$,
i.e., a twisted holomorphic map from $\Sigma$ to $\bM$.

\begin{NB}
    Let us give an intrinsic description. Consider $\Phi$ as a
    $\CC^\times$-equivariant $C^\infty$-map $P_{K_\Sigma^{1/2}}\to
    \bM$. Then $\Phi^*\theta$ is a $C^\infty$ 1-form on
    $P_{K_\Sigma^{1/2}}$. By the equivariance of $\Phi$, it is of
    weight $2$ with respect to the $\CC^\times$-action on
    $P_{K_\Sigma^{1/2}}$.
    Let $1^*_{\CC^\times,K_\Sigma^{1/2}}$ (resp.\
    $1^*_{\CC^\times,\bM}$) be the vector field over
    $P_{K_\Sigma^{1/2}}$ (resp.\@ $\bM$) generated by
    $1\in\operatorname{Lie}\CC^\times$. By the equivariance of $\Phi$,
    $d\Phi(1^*_{\CC^\times,K_\Sigma^{1/2}}) = -1^*_{\CC^\times,\bM}$.
    By our assumption $\theta(1^*_{\CC^\times,\bM})$ is a holomorphic
    function on $\bM$ such that $t^*(\theta(1^*_{\CC^\times,\bM})) =
    t^2 \theta(1^*_{\CC^\times,\bM})$.
    We combine two as $\Phi^*\theta + \Phi^*
    \theta(1^*_{\CC^\times,\bM}) \Omega$, where $\Omega$ is a
    connection on $P_{K_\Sigma^{1/2}}$ whose $(0,1)$-part is the
    holomorphic structure on $K_\Sigma^{1/2}$.

    Then it vanishes on $1^*_{\CC^\times,K_\Sigma^{1/2}}$ and of weight
    $2$. Therefore it is a $K_\Sigma$-valued $1$-form on
    $\Sigma$. Taking $(0,1)$-part
    \(
    \left(\Phi^*\theta + \Phi^*
    \theta(1^*_{\CC^\times,\bM}) \Omega\right)^{(0,1)},
    \)
    and identifying $K_\Sigma\cong\Lambda^{1,0}$, we regard it as a
    $(1,1)$-form. It is independent of the choice of the connection
    $\Omega$ as a different choice is $\Omega + \pi^*\omega$ for a
    $(1,0)$-form $\omega$ on $\Sigma$. Now it is clear that
    \(
        \left(\Phi^*\theta + \Phi^*
    \theta(1^*_{\CC^\times,\bM}) \Omega\right)^{(0,1)}
    = (\Phi^*\theta)^{(0,1)}.
    \)
\end{NB}%

Let $\cF$ be the space of fields, i.e., the space of all
$\CC^\times$-equivariant $C^\infty$-maps $\Phi\colon
P_{K_\Sigma^{1/2}}\to \bM$.

We can consider $\varphi_\CS(\CC_\cF)$, the sheaf of vanishing cycle
with respect to $\CS$ on the moduli space of holomorphic sections of
$K_{\Sigma^{1/2}}\times_{\CC^\times}\bM$ as in
\cite[\S7]{2015arXiv150303676N}.

When $\Sigma$ is an elliptic curve, we do not need to introduce a
$\CC^\times$-action as we have a nonvanishing holomorphic $1$-form
$dz$. We consider a genuine $C^\infty$-map $\Phi\colon \Sigma\to \bM$
and define
\begin{equation*}
    \CS(\Phi) = \int_\Sigma (\Phi^*\theta)^{(0,1)}\wedge dz.
\end{equation*}

\begin{Remark}
Suppose that we have a $\CC^\times\times\CC^\times$-action on $\bM$
such that $(t_1,t_2)^*\theta = t_1 t_2 \theta$. It corresponds to a
\emph{cotangent type} gauge theory in \cite{2015arXiv150303676N}. We
take a $\CC^\times\times\CC^\times$-bundle $P'$ such that the
associated bundle $P'_{\CC^\times\times\CC^\times}\CC\cong K_\Sigma$,
where $\CC^\times\times\CC^\times$ acts on $\CC$ by $(t_1,t_2)z = t_1
t_2 z$. (In other words, we take two line bundles $M_1$, $M_2$ over
$\Sigma$ such that $M_1\otimes M_2 = K_\Sigma$.) Then a
$\CC^\times\times\CC^\times$-equivariant $C^\infty$-map $\Phi\colon
P'\to \bM$ gives a well-defined $(1,1)$-form $(\Phi^*\theta)^{(0,1)}$.
\end{Remark}

\subsection{Gauged $\sigma$-models}

Next suppose the following data are given:
\begin{itemize}
      \item $(\bM,\omega)$ is a (holomorphic) symplectic manifold with
a $G$-action preserving $\omega$.
  \item there is a $\CC^\times$-action commuting with the $G$-action,
such that $t^*\omega = t^2\omega$ for $t\in\CC^\times$.
      \item there is a $G$-invariant Liouville form $\theta$ such that
    $d\theta = \omega$ and $t^*\theta = t^2 \theta$.
\end{itemize}
For a representation, the second $\CC^\times$-action is the scaling
one.
Note that $
\begin{NB}
    \langle\xi,\bmu\rangle \colon 
\end{NB}%
\g\ni \xi\mapsto -\theta(\xi^*)\in C^\infty(\bM)$ is a comoment map
for the $G$-action on $\bM$. Here $\xi^* \equiv \xi^*_\bM$ is the
vector field on $\bM$ generated by $\xi\in\g$.
\begin{NB}
    By Cartan formula, $d(\theta(\xi^*)) = d i_{\xi^*}\theta = -
    i_{\xi^*} d\theta + L_{\xi^*}\theta$. Since $\theta$ is
    $G$-invariant, we have $L_{\xi^*}\theta = 0$. We substitute
    $d\theta = \omega$, we find $-d\theta(\xi^*) = i_{\xi^*}\omega$.
\end{NB}%
\begin{NB}
    Let $1^*_{\CC^\times,x}$ be the generating vector field of the
    $\CC^\times$-action. We have
    $d(\theta(1^*_{\CC^\times,x})) = -i_{1^*_{\CC^\times,x}}\omega
    + L_{1^*_{\CC^\times,x}} \theta = -i_{1^*_{\CC^\times,x}}\omega
    + 2\theta$.

    From the commutativity of $\CC^\times$, $t_*(1^*_{\CC^\times,x})=
    d/ds|_{s=0} t \cdot e^s\cdot x = d/ds|_{s=0} e^s\cdot t \cdot x =
    1^*_{\CC^\times,t\cdot x}$. The condition says
    $\theta(t_*(1^*_{\CC^\times,x})) = t^2
    \theta(1^*_{\CC^\times,x})$. Therefore
    $\theta(1^*_{\CC^\times,t\cdot x}) = t^2 \theta(1^*_{\CC^\times,x})$.
\end{NB}

\begin{NB}
    For a symplectic vector space $(\bM,\omega)$, we define $\theta_x
    = \omega(x,dx)$ for $x\in\bM$. Here $dx$ is the
    tautological $\bM$-valued $1$-form, i.e., $dx(v) = v$, in the left
    hand side $v$ is a tangent vector at $x$, in the right hand side
    $v$ is in $\bM$, and the identification is $T_x\bM \cong \bM$. It
    is $G$-invariant and $d\theta = \omega(dx,dx)$.
\end{NB}%

We now consider a general $G$.
We also fix a $C^\infty$ principal $G$-bundle $P$. We then consider
the fiber product $P\times_\Sigma P_{K_\Sigma^{1/2}}$, which is a
principal $G\times \CC^\times$-bundle over $\Sigma$.

Now a \emph{field} consists of pairs
\begin{itemize}
      \item $\DB+A$ : a partial connection on $P$,
    
      \item $\Phi$ : a $C^\infty$ map $P\times_\Sigma
    P_{K_\Sigma^{1/2}}\to \bM$, which is equivariant under
    $G\times\CC^\times$-action.
\end{itemize}
The space of fields is denoted by $\cF$ again.

We regard $\DB+A$ as a collection of $\g$-valued $(0,1)$-forms
$A_\alpha$ such that $A_\alpha = -\DB g'_{\alpha\beta}\cdot
{g'_{\alpha\beta}}^{-1} + g'_{\alpha\beta}A_\beta
{g'_{\alpha\beta}}^{-1}$, where $g'_{\alpha\beta}$ is the transition
function for $P$. Similarly we regard $\Phi$ as a collection
$\Phi_\alpha\colon U_\alpha\to \bM$ such that $\Phi_\alpha =
(g'_{\alpha\beta},g_{\alpha\beta})\cdot \Phi_\beta$, using the
$G\times\CC^\times$-action on $\bM$.
\begin{NB}
    Thus $\DB \Phi_\alpha = \DB g'_{\alpha\beta} \cdot \Phi_\beta 
    + (g'_{\alpha\beta},g_{\alpha,\beta}) \DB\Phi_\beta$. Therefore
    \begin{equation*}
        \begin{split}
            \langle\theta,\DB\Phi_\alpha\rangle
            & 
            = g_{\alpha\beta}^2 \langle \theta,\DB\Phi_\beta\rangle
        + \langle ((g'_{\alpha\beta},g_{\alpha\beta})\cdot \Phi_\beta)^*\theta,
        \DB g'_{\alpha\beta}\cdot {g'_{\alpha\beta}}^{-1}\rangle
\\
          & 
          = g_{\alpha\beta}^2 \langle \theta,\DB\Phi_\beta\rangle
        + \langle \Phi_\alpha^*\theta,
        \DB g'_{\alpha\beta}\cdot {g'_{\alpha\beta}}^{-1}\rangle
        \end{split}
    \end{equation*}
\end{NB}%
The term involving $\DB g'_{\alpha\beta}\cdot {g'_{\alpha\beta}}^{-1}$
is absorbed by coupling with the moment map $\bmu$ on $\bM$, as
$\langle\xi,\bmu\rangle = -\theta(\xi^*)$. In fact, we have
\begin{equation*}
    \begin{split}
        & \langle A_\alpha, \bmu(\Phi_\alpha)\rangle
\\
= \; &
   -\langle \DB g'_{\alpha\beta} \cdot {g'_{\alpha\beta}}^{-1},
   \bmu(\Phi_\alpha)\rangle
   + \langle g'_{\alpha\beta} A_\beta {g'_{\alpha\beta}}^{-1},
   \bmu((g'_{\alpha\beta},g_{\alpha\beta})\cdot \Phi_\beta)\rangle
\\
= \; &
   \langle 
   \Phi_\alpha^*\theta,
        \DB g'_{\alpha\beta}\cdot {g'_{\alpha\beta}}^{-1}\rangle
   + 
   g_{\alpha\beta}^2 \langle A_\beta,
   \bmu(\Phi_\beta)\rangle,
    \end{split}
\end{equation*}
where we have used the equivariance of the moment map in the second
equality. Thus $\{ (\Phi_\alpha^*\theta)^{(0,1)} - \langle A_\alpha,
\bmu(\Phi_\alpha)\rangle\}$ defines a well-defined $(1,1)$-form. We
introduce the holomorphic Chern-Simons functional as
\begin{equation*}
    \CS(A,\Phi) \defeq \int_\Sigma (\Phi_\alpha^*\theta)^{(0,1)}
    -\langle A_\alpha,\bmu(\Phi_\alpha)\rangle.
\end{equation*}
Then $(A,\Phi)$ is a critical point of $\CS$ if and only if
\begin{itemize}
      \item $\bmu(\Phi_\alpha) = 0$,
    \item $\{ \Phi_\alpha\}$ is a holomorphic section of
  $(P\times_\Sigma P_{K_\Sigma^{1/2}})\times_{G\times\CC^\times}\bM$. Here
  a holomorphic structure on $P$ is given by $\DB+A$.
\end{itemize}

Thus $(A,\Phi)$ is a twisted holomorphic map from $\Sigma$ to the
quotient stack $[\bmu^{-1}(0)/ G]$.

When $\bM$ is a symplectic representation of $\bM$, this construction
recovers the previous one in \cite[\S7]{2015arXiv150303676N}.

We consider the dual of the equivariant cohomology with compact
support of $\cF$ with vanishing cycle coefficient:
\begin{equation}\label{eq:7}
    H^*_{c,\mathcal G(P)}(\cF, \varphi_\CS(\CC_\cF))^*,
\end{equation}
where $\mathcal G(P)$ is the complex gauge group denoted by $\mathcal
G_\CC(P)$ in \cite{2015arXiv150303676N}.

The main proposal in \cite{2015arXiv150303676N} can be generalized in
a straightforward manner:
\begin{itemize}
      \item Take $\Sigma = \proj^1$.
    Define a commutative product on $H^*_{c,\mathcal
      G(P)}(\cF, \varphi_\CS(\CC_\cF))^*$ and consider an affine
    variety given by its spectrum. It should be the underlying affine
    variety of the moduli space of vacua of the gauged $\sigma$-model
    for $(G,\bM,\omega)$ \cite{HKLR}. In particular, it is expected to
    satisfy various properties claimed in the physics literature.
\end{itemize}

\section{Symplectic leaves and transversal slices}

Let us consider the case $\bM$ is a symplectic representation of
$G$. 

\begin{Question}\label{q:leaves}
    (1) Are there only finitely many symplectic leaves in the Coulomb
    branch $\mathcal M_C$ ?

    (2) Are symplectic leaves and transversal slices Coulomb branches
    of gauge theories ?

    (3) Is there a `natural' order-reversing bijection between
    symplectic leaves of the Coulomb branch $\mathcal M_C$ and the
    Higgs branch $\mathcal M_H$ ?
    Suppose (2) is true. Are gauge theories for symplectic leaves and
    transversal slices of $\mathcal M_C$ and $\mathcal M_H$
    interchanged under the bijection ?
\end{Question}

For quiver gauge theories of type $ADE$, Higgs branches are quiver
varieties while Coulomb branches are slices in affine Grassmannian.
See \cite[\S3(ii)]{2015arXiv150303676N}. (This statement is true only
the \emph{dominance} condition is satisfied. See
\subsecref{subsec:nonzeroW}.) In this case it is well-known that both
leaves are parametrized by dominant weights, and the answer is
yes. See \subsecref{subsec:nonzeroW} for detail.

Also the question~(3) is one of requirements in the definition of the
symplectic duality \cite[\S10]{2014arXiv1407.0964B}.\footnote{The
  author thanks Tom Braden who explained the statement in talks by at
  Kyoto and Boston.}
Properties (1),(2) are known for Higgs branches $\mathcal M_H$, as we
will review below. Therefore (1),(2) for $\mathcal M_C$ are natural to
ask. The second question in (3) \emph{does} make sense thanks to this
fact.

In fact, the symplectic duality requires an order-reversing bijection
on posets of \emph{special} symplectic leaves, but we ignore a subtle
difference between special and arbitrary symplectic leaves at this
moment. Properties (1),(2) make sense even without relation to $\mathcal M_H$.

The naturality in (2) is vague, but the author cannot make it
precise. In fact, we will find examples where (3) fail in
\subsecref{sec:affine} below. It means that an ad-hoc bijection for
(2) may be unnatural. Therefore we should regard (2) are also false
for these. These examples will be related to \emph{un}special leaves,
hence this question must be corrected in future.

Also the question (3) is a little imprecise as Coulomb branches are
affine algebraic varieties, and leaves are usually only
quasi-affine. A typical example is the $n^{\mathrm{th}}$ symmetric
product $S^n\CC^2$ and its standard stratification $\bigsqcup
S_\nu\CC^2$ parametrized by partitions $\nu$ of $n$. The stratum
$S_\nu\CC^2$ is an open subvariety in $S^{\nu_1}\CC^2\times
S^{\nu_2}\CC^2\times\cdots$ for $\nu = (1^{\nu_1}2^{\nu_2}\dots)$. The
latter is the Coulomb branch of a certain gauge theory. See
\subsecref{sec:affine} below. Note also that
$S^{\nu_1}\CC^2\times S^{\nu_2}\CC^2\times\cdots \to
\overline{S_\nu\CC^2}$ is a finite birational morphism, but is
\emph{not} an isomorphism. Therefore, as the best, we can ask whether
the closure of a symplectic leaf is an image of a finite birational
morphism from the Coulomb branch of a gauge theory. It is not clear
(at least to the author) whether this determine the gauge theory
uniquely or not. We ignore this complexity and simply ask whether a
symplectic leaf is the Coulomb branch hereafter.

An important special case of (2) is
\begin{Question}\label{q:smooth}
    Suppose the Higgs branch $\mathcal M_H = \bM\tslash G$ consists of
    only $\bM^G$.\footnote{The author thanks Alexander Braverman to
      point out a mistake in an earlier version.} Is the corresponding
    Coulomb branch $\mathcal M_C$ a smooth symplectic manifold ? Is
    converse true also ?
\end{Question}

\subsection{Higgs branch}\label{sec:higgs-branch}

Let us recall a well-known result on symplecitc leaves on the Higgs
branch $\bM\tslash G = \bmu^{-1}(0)\dslash G$. 

We consider $\bM\tslash G$ as a set of closed $G$-orbits in
$\bmu^{-1}(0)$. It has a natural stratification given by conjugacy
classes of stabilizers (see \cite[\S6]{Na-quiver}, which is based on
\cite{MR1127479}):
\begin{equation}\label{eq:3}
    \bM\tslash G = \bigsqcup_{(\widehat G)} (\bM\tslash G)_{(\widehat G)},
\end{equation}
where a stratum $(\bM\tslash G)_{(\widehat G)}$ consists of orbits through points
$x\in\bM\tslash G$ whose stablizer $\Stab_G(x)$ is conjugate to $\hat
G$. Note that $\widehat G$ is a reductive group as $x$ has a closed orbit.

Let us describe $(\bM\tslash G)_{(\widehat G)}$ as a symplectic
reduction.  (See \cite[Th.~3.5]{MR1127479} for detail.)
Let $\bM^{\widehat G} = \{ m\in\bM \mid \text{$gm = m$ for $g\in \hat
  G$}\}$ denote the fixed point locus. It is a symplectic vector subspace.
Let $N_G(\widehat G)$ denote the normalizer of $\widehat G$ in $G$. We have an
action of $\nicefrac{N_G(\widehat G)}{\widehat G}$ on $\bM^{\widehat G}$. Then we have
\begin{equation}\label{eq:1}
    (\bM\tslash G)_{(\widehat G)} = (\bM^{\widehat G}\tslash
    \nicefrac{N_G(\widehat G)}{\widehat G})_{(\{e\})}.
\end{equation}
Here the subscript $(\{e\})$ means $\nicefrac{N_G(\widehat G)}{\hat
  G}$-orbits through points whose stabilizer is conjugate to the
identity group $\{e\}$, i.e., free orbits. In particular,
$(\bM^{\widehat G}\tslash \nicefrac{N_G(\widehat G)}{\widehat
  G})_{(\{e\})}$ is symplectic. We do \emph{not} claim that it is a
symplectic leaf as it may not be connected in general. This is a
subtle point, which we do not understand well. See the
example~\ref{ex:SL2}.

A transversal slice to a stratum $(\bM\tslash G)_{\widehat G}$ is also
described as a symplectic reduction. (See \cite[\S2]{MR1127479}. See
also \cite[\S3]{Na-qaff} and \cite{CB:normal}.) Let $m\in
\bmu^{-1}(0)$ such that $\Stab_G(m) = \widehat G$.
We take the orbit $Gm$ through $m$ and its tangent space $T_mGm$ at
$m$. The latter is an isotropic subspace. We consider the symplectic
normal space $\widehat\bM\defeq (T_m Gm)^\omega/T_m Gm$, which has a natural symplectic structure. It is naturally a representation of $\widehat G$. Then
$\bM\tslash G$ and $\widehat\bM\tslash \widehat G$ are locally isomorphic
around $[m]$ and $[0]$.
\begin{NB}
    $\bM$ is locally look like $G\times_{\hat
      G}(\hat\g^{\perp*}\times\widehat\bM)$ so that the moment map is
    $\mu(G\cdot(g,\xi,m)) = \operatorname{Ad}^*(g)(\xi + \hat\mu(m))$.
\end{NB}%
Under the local isomorphism the stratum $(\bM\tslash G)_{(\widehat
  G)}$ is mapped to $(\widehat\bM\tslash \widehat G)_{(\widehat G)}$.
This is nothing but the fixed subspace $T \defeq \widehat\bM^{\widehat G}$.
We take a complementary subspace $T^\perp$ of $T$ in $\widehat\bM$. The transversal slice is given by
\begin{equation}\label{eq:6}
    T^\perp\tslash \widehat G.
\end{equation}

\subsection{Coulomb branch}

Combining Question~\ref{q:leaves} with (\ref{eq:1}, \ref{eq:6}), we
arrive at the following:

\begin{Question}
    Strata (resp.\@ transversal slices) of $\mathcal M_C$ are the
    Coulomb branches of gauge theories $\Hyp(T^\perp)\tslabar \widehat
    G$ (resp.\@ $\Hyp(\bM^{\widehat G})\tslabar \nicefrac{N_G(\widehat
      G)}{\widehat G}$) ?
\end{Question}

We have the following naive construction: Consider morphisms between
Higgs branches.
\begin{equation*}
    \bM^{\widehat G}\tslash \nicefrac{N_G(\widehat
      G)}{\widehat G}
    \to \bM\tslash G \to
    T^\perp \tslash \widehat G.
\end{equation*}
They induces morphisms between moduli spaces of twisted holomorphic
maps from a Riemann surface $\Sigma$. Taking the dual of cohomology
groups with vanishing cycle coefficients, we get homomorphisms in the
same direction. They conjecturally respect the multiplication, and
hence induce morphisms on the spectrum. Thus morphisms between Coulomb
branches goes in the opposite direction, and the strata and slices are
interchanged.

In the proposed `definition' of Coulomb branches, the affine GIT
quotient $\bM\tslash G = \bmu^{-1}(0)\dslash G$ is not enough: we need
the quotient stack $[\bmu^{-1}(0)/G]$. But we do not think this makes
a crucial difference.

\section{Complete intersection and
conical $\CC^\times$-action}

For a coweight $\la$ of $G$, we define
\begin{equation}\label{eq:9}
    \Delta(\la) \defeq -\sum_{\alpha\in\Delta^+}|\langle\alpha,\la\rangle|
    + \frac14 \sum_\mu |\langle\mu,\la\rangle| \dim \bM(\mu),
\end{equation}
where $\Delta^+$ is the set of positive roots, and $\mu$ runs over the
set of weights of $\bM$. Here $\bM(\mu)$ is the weight space.

The gauge theory $\Hyp(\bM)\tslabar G$ is said \emph{good} if
$2\Delta(\la)> 1$ for any $\la\neq 0$, and \emph{ugly} if
$2\Delta(\la)\ge 1$ and not good. 
(Note $2\Delta(\la)\in\ZZ$ as weights appear in pairs for a symplectic representation.)
The monopole formula predicts that $2\Delta(\la)$ corresponds to a
weight of the $\CC^\times$-action on $\mathcal M_C$. (See
\cite[\S4]{2015arXiv150303676N}.) Therefore
\begin{itemize}
      \item $\Hyp(\bM)\tslabar G$ is `good' or `ugly' if and only if
    the $\CC^\times$-action on the Coulomb branch $\mathcal M_C$ is
    \emph{conical}.
\end{itemize}
Recall that a $\CC^\times$-action is conical if the coordinate ring
$\CC[\mathcal M_C]$ is only with nonnegative weights, and the zero
weight space is $1$-dimensional consisting of constant functions.
Similarly $\Hyp(\bM)\tslabar G$ is good if and only if the
$\CC^\times$-aciton is conical and there is no function of weight $1$.

On the other hand, it is observed in
\cite[\S2(iv)]{2015arXiv150303676N} that $2\Delta(\la)\ge 1$ looks
similar to a complete intersection criterion of $\bmu=0$ for the Higgs
branch by Crawley-Boevey \cite{CB} for quiver gauge theories. The
following question was asked:

\begin{Question}\label{q:complete}
    Is the followings true ?

    The $\CC^\times$-action of the Coulomb branch $\mathcal M_C$ is
    conical
    \begin{NB}
        (resp.\@ and has no function of weight $1$) 
    \end{NB}%
    \begin{NB}
        if and 
    \end{NB}%
    only if the level set $\bmu^{-1}(0)$ of the moment map equation
    for the Higgs branch $\mathcal M_H$ is complete intersection
    \begin{NB}
        (resp.\@ and irreducible)
    \end{NB}%
    in $\bM$ of $\dim = \dim \bM - \dim G$.

    We assume $\bM$ is a faithful representation of $G$.
\end{Question}

Since $\bM\tslash G$ depends on the image $G\to \grpSp(\bM)$, the
faithfulness assumption is a natural requirement. 
If $G\to \grpSp(\bM)$ has a positive dimensional kernel,
$\bmu^{-1}(0)$ has dimension larger than $\dim \bM - \dim G$.
If $G\to\grpSp(\bM)$ has only finite kernel, the complete intersection
property does not matter, but we have the following example:
Suppose that $\CC^\times$ acts on $\bN=\CC$ by weight $N\neq 0$, and
take $\bM = \bN\oplus\bN^*$. The level set $\bmu^{-1}(0)$ is the same
for the case $N=1$ and is not irreducible, but it is good if $N>1$ and
ugly if $N=1$. On the other hand, the Coulomb branch $\mathcal M_C$
\emph{does} depends on $N$. One can see from the monopole formula as
$\mathcal M_C = \CC^2/(\ZZ/N\ZZ)$ with the $\CC^\times$-action induced
from $t\cdot (x,y) = (tx,ty)$ for $(x,y)\in\CC^2$.

In \cite{2015arXiv150303676N} it was asked two conditions in
Question~\ref{q:complete} are equivalent. But it turns out that there
are counter examples for the converse.

\section{Poisson and intersection homology groups}

Let $X$ be a Poisson variety and let $H\!P_*(X)$ denote its Poisson
homology group, defined in \cite{MR2729282}. We are interested in the
case $X$ is affine and the degree $0$ part $H\!P_0(X)$. It is known
that $H\!P_0(X)$ is the quotient of $\CC[X]$ by the linear span of all
brackets. The following is a $\mathcal M_H/\mathcal M_C$ version of
conjecture of Proudfoot \cite{MR3238150}:

\begin{Question}\label{q:HPHH}
    Suppose $\Hyp(\bM)\tslabar G$ is good.

    \textup{(1)}
    Are there natural isomorphisms of graded vector spaces ?
    \begin{equation*}
        H\! P_0(\mathcal M_H) \cong I\!H^*(\mathcal M_C), \quad
        H\!P_0(\mathcal M_C) \cong I\!H^*(\mathcal M_H).
    \end{equation*}
    Here the grading of the left hand sides are given by the
    $\CC^\times$-action.

    \textup{(2)} Let $\cAh[\mathcal M_H]$, $\cAh[\mathcal M_C]$ be the
    quantization of $\CC[\mathcal M_H]$, $\CC[\mathcal M_C]$ respectively.
    Are there natural isomorphisms of graded vector spaces ?
    \begin{equation*}
            H\!H_0(\cAh[\mathcal M_H])
        \cong I\!H^*_{\CC^\times}(\mathcal M_C), \quad
        H\!H_0(\cAh[\mathcal M_C])
        \cong I\!H^*_{\CC^\times}(\mathcal M_H).
    \end{equation*}
    Here $H\!H_0$ denote the zero-th Hochschild homology group, i.e.,
    the quotient of the quantization by the span of its commutator.
\end{Question}

The quantization $\cAh[\mathcal M_H]$ is defined under the assumption
that $\bM$ is of cotangent type, i.e., $\bM = \bN\oplus\bN^*$ for a
$G$-module $\bN$. Then we consider the ring $\mathcal D_\hbar(\bN)$ of
$\hbar$-differential operators on $\bN$, and define $\cAh[\mathcal
M_H]$ as its quantum hamiltonian reduction by $G$. The Coulomb branch
$\mathcal M_C$ is expected to have a quantization always as a
$\CC^\times$-equivariant homology group in \eqref{eq:7}, where
$\CC^\times$ acts on $\cF$ through the $\CC^\times$-action on $\Sigma
= \proj^1$. Here $\hbar$ appears as $H^*_{\CC^\times}(\mathrm{pt}) =
\CC[\hbar]$.

This question is well formulated \emph{except} the meaning of the
naturality, which we do not discuss here. If we do not assume the good
condition, we easily find couter-examples anyway, as we have many
cases with $\mathcal M_H = \{0\}$, while $\mathcal M_C$ is
nontrivial. In this case, $\mathcal M_C$ is expected to be smooth,
hence $I\!H^*(\mathcal M_C) \cong H^*(\mathcal M_C)$, $H\!P_0(\mathcal
M_C) \cong H^{\dim \mathcal M_C}(\mathcal M_C)$. These cohomology
groups are often nontrivial.

In order to exclude these cases, we have assumed the good 
condition.



\section{(Counter)Examples}

\subsection{Quiver gauge theories of type $ADE$}\label{sec:quiv-gauge-theor}

Let us consider a quiver gauge theory (see
\cite[2(iv)]{2015arXiv150303676N}). We follow the notation there. We
first suppose $W=0$. Note that scalars in $G = \prod_{i\in Q_0}
\GL(V_i)$ act trivially on $\bM$. We need $\bM$ is a faithful
representation in Question~\ref{q:complete} for the Higgs branch.
Also the embedding $\la\colon\CC^\times\to G$ of scalars gives
$\Delta(\la) = 0$. Therefore $\mathcal M_C$ is never conical.
Therefore we replace $G$ by $\prod_{i\in Q_0} \GL(V_i)/\CC^\times$.
We also assume that the support $\{ i\in Q_0 \mid \dim V_i\neq 0\}$ of
$\dim V$ is connected, as $\bM$ is not faithful otherwise.
\begin{NB}
    For example, a smaller graph could be disconnected. Then we need
    to divide $(\CC^\times)^{\pi_0(Q)}$.
\end{NB}%
We write the dimension vector $\dim V$ by $\mathbf v$.

We first assume that the quiver $Q = (Q_0,Q_1)$ is of type $ADE$.
The Coulomb branch $\mathcal M_C$ is the moduli space of centered
$ADE$ monopoles on $\RR^3$ where $\mathbf v$ is the corresponding
monopole charge (\cite{MR1451054} for type $A$, \cite{MR1677752} in
general\footnote{The author thanks Vasily Pestun who explained this
  statement to a collaborator of \cite{BFN}. It is compatible with his
  work with Nekrasov of $4d$ quiver gauge theories
  \cite{2012arXiv1211.2240N}.}). This is a smooth symplectic
manifold. By \cite{MR769355,MR987771,MR1625475}, it is the same as the
moduli space of centered based rational maps from $\proj^1$ to the
flag manifold of type $ADE$. The definition in \cite{BFN} produces
$\mathcal M_C$ as this moduli space.

On the other hand, consider the Higgs branch $\mathcal M_H$.  It is
known that any element in $\bmu^{-1}(0)$ is automatically nilpotent
for $ADE$ quivers \cite{Lu-crystal}. Therefore the only closed orbit
in $\bmu^{-1}(0)$ is $0$. (See also \cite[Prop.~6.7]{Na-quiver}.)
Therefore the answer to Question~\ref{q:smooth} is yes.

\begin{Proposition}
    $\bmu^{-1}(0)$ is complete intersection of dimension $\dim \bM - \dim
    G$ if and only if $\mathbf v$ is a positive root.
\end{Proposition}

\begin{proof}
By the criterion in \cite[Th.~1.1]{CB}, $\bmu^{-1}(0)$ is complete
intersection of $\dim = \dim \bM - \dim G$ if and only if
\begin{equation}\label{eq:8}
    2 - \langle\mathbf v, \mathbf C \mathbf v\rangle
    \ge \sum_k (2 - \langle \beta^{(k)},\mathbf C \beta^{(k)}\rangle)
\end{equation}
for any decomposition $\mathbf v = \sum \beta^{(k)}$ such that
$\beta^{(k)}$ is a positive root (or equivalently nonzero positive
vector).
Here $\mathbf C = (2\delta_{ij} - a_{ij})$ with $a_{ij}$, the number
of edges (regardless of orientation) between $i$ and $j$ if $i\neq j$,
and its twice if $i=j$.  (The latter does not occur for type $ADE$.)
Since we are assuming $Q$ is of type $ADE$, the right hand side is
always $0$. On the other hand, the left hand side is nonpositive, and
is zero if and only if $\mathbf v$ is a positive root.
\end{proof}

\begin{NB}
    \eqref{eq:8} is equivalent to
    \begin{equation*}
        -\sum_{k < l} \langle \beta^{(k)}, \mathbf C\beta^{(l)}\rangle
        \ge \# \{ \beta^{(k)} \} - 1.
    \end{equation*}
\end{NB}

Next let us study $\Delta(\la)$ in \eqref{eq:9}.
\begin{NB}
    Take $\la$, the embedding of one of entries of $T(V_i)$ for a
    vertex $i$. Then $2\Delta(\la) = - 2(\dim V_i - 1) + \sum a_{ij}
    \dim V_j$. Therefore the theory is good or ugly only if \( - \sum
    (2\delta_{ij} - a_{ij}) \dim V_j \ge -1.  \) See
    \cite[(2.2)]{2015arXiv150303676N}.
\end{NB}%
\begin{NB}
    But it does not imply $\mathbf v$ is a root yet. Consider
    $\mathbf v = (1,2,2,1)$ for type $A_4$.
\end{NB}%
\begin{NB}
    Next suppose $\dim V_i \ge 2$, and take $\la$ so that $\la^i =
    (1,-1,0,\dots)$. Then $\Delta(\la) = - 4 - 4 (\dim V_i - 2) +
    2\sum a_{ij} \dim V_j = 2\left\{ -2 (\dim V_i - 1) + \sum a_{ij}
        \dim V_j\right\}$.
\end{NB}%
We take a maximal torus $T = \prod_i T(V_i)/\CC^\times$, where
$T(V_i)$ is the diagonal subgroup of $\GL(V_i)$ and consider
$\la\colon\CC^\times\to T$.
According to weights on $\bigoplus_i V_i$, we decompose $V = V^{(1)}\oplus V^{(2)}\oplus\cdots$ such that $\la(t)$ acts on $t^{\la^{(k)}}$ on $V^{(k)}$.
%
Set $\beta^{k} = (\beta^{(k)}_i) = (\dim V^{(k)}_i)$. Then
\begin{equation}\label{eq:10}
    \begin{split}
    2\Delta(\la) &= \sum_{k<l} |\la^{(k)} - \la^{(l)}| \left( \sum_i
    -2\beta^{(k)}_i \beta^{(l)}_i
    + \sum_{i,j} a_{ij} \beta^{(k)}_i \beta^{(l)}_j
    \right)
\\
    & = -\sum_{k<l} |\la^{(k)} - \la^{(l)}| \langle \beta^{(k)}, \mathbf C
    \beta^{(l)}\rangle.
    \end{split}
\end{equation}
\begin{NB}
    Since $\langle\bullet,\mathbf C\bullet\rangle$ is positive
    definite, Schwarz' inequality says $|\langle\beta^{(k)},\mathbf
    C\beta^{(l)}\rangle| < |\beta^{(k)}||\beta^{(l)}| = 2$ unless
    $\beta^{(k)} = \beta^{(l)}$. Therefore $\langle\beta^{(k)},
    \mathbf C\beta^{(l)}\rangle = \pm 1, 0$. It is $+1\Longleftrightarrow
    \beta^{(k)} - \beta^{(l)}$ is a root. It is $-1\Longleftrightarrow
    \beta^{(k)} + \beta^{(l)}$ is a root.
\end{NB}

Suppose that $\mathbf v$ is not a positive root. We decompose $\mathbf v$ as
sum $\sum \beta^{(k)}$ of positive roots $\beta^{(k)}$ so that any
combination $\beta^{(k)} + \beta^{(l)}$ is not a root. Then $4\le
\langle \beta^{(k)} + \beta^{(l)}, \mathbf
C(\beta^{(k)}+\beta^{(l)})\rangle = \langle \beta^{(k)},\mathbf
C\beta^{(k)}\rangle + \langle \beta^{(l)}, \mathbf C\beta^{(l)}\rangle
+ 2\langle \beta^{(k)}, \mathbf C\beta^{(l)}\rangle = 4 + 2\langle
\beta^{(k)}, \mathbf C\beta^{(l)}\rangle$ for $k\neq l$. Therefore
$\langle\beta^{(k)}, \mathbf C\beta^{(l)}\rangle\ge 0$.
Hence $2\Delta(\la)\le 0$, so the gauge theory is not good or ugly.
Therefore the answer to Question~\ref{q:complete} is yes.

Let us continue to study $\Delta(\la)$.

\begin{Proposition}\label{prop:good}
$\Hyp(\bM)\tslabar G$ is never good.

\end{Proposition}

\begin{proof}
    Take a decomposition $\mathbf v = \beta^1 + \beta^2$, $\beta^1 =
    \mathbf v - \alpha_i$, $\beta^2 = \alpha_i$, we find $2\Delta(\la)
    = -|\la^{(1)}-\la^{(2)}|(\langle \alpha_i, \mathbf C\mathbf v\rangle
    - 2)$.  We have $2\Delta(\la) > 1$ for $\la$ of this form if and
    only if $\langle \alpha_i, \mathbf C\mathbf v\rangle\le 0$.
But this is never possible unless $\mathbf v=0$.
\begin{NB}
    As $\langle \mathbf v, \mathbf C\mathbf v\rangle \ge 0$, and the
    equality if and only $\mathbf v = 0$.
\end{NB}%
\begin{NB}
    For example, of type $A$, we take $i$ be the left-most vertex in
    the support of $\mathbf v$.
\end{NB}%
\begin{NB}
(2) We represent a coweight $\la$ as an integer vector
$(\la^i_k)_{i\in Q_0, k=1,\dots,\dim V_i}$. Then
\begin{equation*}
    2\Delta(\la) = -\sum_{i\in Q_0} \sum_{k \neq l} |\la^i_k - \la^i_l|
    + \frac12 \sum_{i,j} a_{ij} \sum_{k,p} | \la^i_k - \la^j_p|.
\end{equation*}
Fixing $i$, $j$, $p$, we have
\begin{equation*}
    \sum_{k\neq l}
    |\la^i_k - \la^i_l| \le 
    2(\dim V_i - 1)
    \sum_{k} |\la^i_k - \la^j_p|.
\end{equation*}
\end{NB}%
\end{proof}

Suppose $Q$ is of type $A$ and $\mathbf v$ is a positive root. Then $G$
is a torus and $2\Delta(\la) > 0$ unless $\la = 0$. Therefore
$\Hyp(\bM)\tslabar G$ is ugly. It is also possible to check
$\Hyp(\bM)\tslabar G$ is ugly if $Q$ is of type $D$ and $\mathbf v$ is
positive root as follows.

Suppose $Q$ is of type $D_\ell$. We consider the case $\mathbf v =
(12\dots 2\genfrac{}{}{0pt}{0}{1}{1})$, other cases are similar. We
represent a coweight $\la$ as an integer vector $(\la^1,\la^2_1,\la^2_2,\dots,
\linebreak[3]\la^{\ell-2}_1,\la^{\ell-2}_2,\la^{\ell-1},\la^\ell)$. We have
\begin{equation*}
    \begin{split}
        2\Delta(\la) & = -2\sum_{p=2}^{\ell-2} |\la^p_1-\la^p_2|
    + |\la^1 - \la^2_1| + |\la^1 - \la^2_2|
    + \sum_{p=2}^{\ell-3} \sum_{a,b=1}^n |\la^p_a - \la^{p+1}_b|
    \\
    & + \sum_{p=\ell-1,\ell} |\la^{\ell-2}_1 - \la^p| + 
    |\la^{\ell-2}_2 - \la^p|.
    \end{split}
\end{equation*}
We use
\begin{equation*}
    \begin{split}
        & 2 |\la^2_1 - \la^2_2| \le |\la^1 - \la^2_1| + |\la^1 -
        \la^2_2| + \frac12 \sum_{a,b=1}^2 |\la^2_a - \la^3_b|,
        \\
        & 2 |\la^p_1 - \la^p_2| \le \frac12 \sum_{a,b=1}^2
        |\la^{p-1}_a - \la^{p}_b| + \frac12 \sum_{a,b=1}^2 |\la^{p}_a
        - \la^{p+1}_b| \quad (p=3,\dots,\ell-3),
        \\
        & 2 |\la^{\ell-2}_1 - \la^{\ell-2}_2| \le \frac12 \sum_{a,b=1}^2
        |\la^{\ell-3}_a - \la^{\ell-2}_b| + 
        \sum_{p=\ell-1,\ell} |\la^{\ell-2}_1 - \la^p| + 
    |\la^{\ell-2}_2 - \la^p|.
    \end{split}
\end{equation*}
(In fact, $p=\ell-1$ is enough in the second sum in the last
equality.)  Taking sum, we find $2\Delta(\la)\ge 0$ and the equality
holds if and only if all entries of $\la$ are the same, i.e., it is
zero as a coweight of $\prod_i \GL(V_i)/\CC^\times$. Therefore it is
ugly.

We do not know for exceptional cases, though it is a finite check.

\subsection{Affine types}\label{sec:affine}
Next suppose the underlying graph is of type \emph{affine} $ADE$ with
$W=0$. In this case, the Coulomb branch $\mathcal M_C$ is
conjecturally the moduli space of calorons, in other words, instantons
on $\RR^3\times S^1$. The charge is again given by $\mathbf v$. More
precisely we need two modifications:
a) We take the symplectic reduction by $\CC^\times$, the action induced from $\CC^\times$-action on the base $\RR^3\times S^1 = \CC^\times \times \CC$.
b) We take the Uhlenbeck partial compactification like in case of
$\RR^4$. Therefore $\mathcal M_C$ has a stratification
\begin{equation}\label{eq:4}
    \mathcal M_C(\mathbf v) = \bigsqcup 
    \mathcal M_C^\circ (\mathbf v - |\nu|\delta)\times
    S_{\nu} (\RR^3\times S^1)_c,
\end{equation}
where $\delta$ is the primitive imaginary root vector, $\mathcal
M_C^\circ (\mathbf v - |\nu|\delta)$ denotes the moduli space of
\emph{genuine} calorons with charge $\mathbf v - |\nu|\delta$, and
$S_\nu(\RR^3\times S^1)$ is a stratum of the symmetric product of
$\RR^3\times S^1$ given by a partition $\nu$, and $S_\nu(\RR^3\times
S^1)_c$ means the symplectic reduction by the
$\CC^\times$-action. Therefore the strata are parametrized by
partitions $\nu$ such that $\mathbf v - |\nu|\delta$ is nonnegative.

This result remains true for Jordan quiver. It corresponds to
$\U(1)$-calorons, and there is no genuine calorons. Therefore the
Coulomb branch is expected to be $S^{\dim V}(\RR^3\times S^1)_c =
\bigsqcup S_\nu(\RR^3\times S^1)_c$. (It is confirmed in \cite{BFN}.)

Moreover 
\begin{aenume}
      \item\label{Cs} The stratum in \eqref{eq:4}, replaced by
    $\mathcal M_C(\mathbf v-|\nu|\delta)\times \prod
    S^{\nu_k}(\RR^3\times S^1)_c$ for $\nu = (1^{\nu_1}
    2^{\nu_2}\dots)$, is the Coulomb branch of the quiver gauge
    theory, where the graph is the disjoint union of the original $Q$
    and copies of the Jordan quiver with the dimension vector $\mathbf v
    - |\nu|\delta$, $\nu_1$, $\nu_2$, \dots respectively.
    \begin{NB}
        Groups are $G_V/\CC^\times$, $\PGL(\nu_1)$, $\PGL(\nu_2)$,\dots.
    \end{NB}%

      \item\label{Ct} The transversal slice is the product of
    \emph{centered} $ADE$ instanton moduli spaces on $\RR^4$ with
    instanton numbers $\nu_1$, $\nu_2$, \dots. If we ignore the
    \emph{centered} condition, it is the Coulomb branch of the quiver
    gauge theory, where the quiver is the disjoint union of copies of
    the original $Q$
    \begin{NB}
        of affine $ADE$ type 
    \end{NB}%
    with the dimension vectors $(\mathbf v, \mathbf w) =
    \underbrace{(\delta,\delta_{0i}),\dots}_{\text{$\nu_1$ times}},
    \underbrace{(2\delta,\delta_{0i}),\dots}_{\text{$\nu_2$ times}}$,
    \dots. See \cite[1(vii) and 3(i)]{2015arXiv150303676N}.
\end{aenume}

\begin{NB}
    Suppose $\mathbf v = n\delta$ for brevity. The open stratum is $\nu =
    \emptyset$. The transversal slice is a point. The closed stratum
    is $\nu = (n)$. The transversal slice is $n$ instantons on $\RR^4$.

    For the stratum $\nu = (1)$ (i.e., $1$ bubble), the transversal
    slice is the moduli of centered $1$-instanton on $\RR^4$.  (e.g.,
    $\RR^4\times (\RR^4/\pm 1) \subset (\RR^4\times \{0\}/\pm 1)$ is
    a stratum. The centered moduli is $\RR^4/\pm 1$.)

    For the stratum $\nu = (1^n)$, i.e., $1$ bubble at $n$ distinct
    points, the closure of the stratum is $S^n(\RR^3\times S^1)$, the
    transversal slice is $(\text{$1$-instantons on $\RR^4$})^n$.
\end{NB}%

\begin{Remark}
    For affine type $A$, the moduli space of calorons is isomorphic to
    the moduli space of framed locally free parabolic sheaves on
    $\proj^1\times\proj^1$. (See \cite{MR2681686,MR2395473} and
    \cite{Takayama}.) This is expected to be true for any affine type
    if we replace locally free parabolic sheaves by principal
    $G_{ADE}$-bundles with parabolic structures. The definition of
    \cite{BFN} gives this moduli space for affine type $A$, and
    conjecturally in general.
\end{Remark}

Let us consider $\mathcal M_H = \bmu^{-1}(0)\dslash G$. A closed
$G$-orbit corresponds to a semisimple representation of the
preprojective algebra of type affine $ADE$. By \cite{CB}, we have the
classification of simple representations: they are either $S_i$ (one
dimensional at the vertex $i$) or have dimension $\delta$. In the
latter case, $\mathcal M_H$ is a particular case of Kronheimer's
construction \cite{Kr}, and is $\RR^4/\Gamma$ for a finite subgroup
$\Gamma\subset\SL(2)$ corresponding to the affine $ADE$ graph. 
The origin corresponds to a direct sum of $S_i$'s, but any point
except the origin is a simple representation.
Since semisimple representation is a direct sum of simple
representations, we have a stratification
\begin{equation}\label{eq:5}
    \mathcal M_H(\mathbf v) = \bigsqcup S_\nu(\RR^4\setminus \{0\}/\Gamma),
\end{equation}
where $\nu$ is a partition such that $\mathbf v - |\nu|\delta$ is
nonnegative. Here we have a direct sum of $S_i$'s corresponding to
$\mathbf v - |\nu|\delta$. It corresponds to the `origin' in $\RR^4$.

A semisimple representation of the preprojective algebra for the
Jordan quiver is a pair of commuting semisimple matrices. Hence
$\bmu^{-1}(0)\dslash G = S^n(\CC^2) = \bigsqcup
S_\nu(\RR^4)$. Therefore we have the same result for the Jordan quiver.

\begin{aenume}
   \setcounter{enumi}{2}
     \item\label{Hs} The stratum in \eqref{eq:5}, replaced by
   $S^{\nu_1}(\RR^4/\Gamma)\times
   S^{\nu_2}(\RR^4/\Gamma)\times\cdots$, is the Higgs branch of the
   quiver gauge theory in (\ref{Ct}) above. The quiver is the disjoint
   union of copies of $Q$, the dimension vectors are $(\mathbf v,\mathbf
   w) = (\nu_1\delta,\delta_{0i}), (\nu_2\delta,\delta_{0i}),\dots$.

     \item\label{Ht} The transversal slice is the product 
   \[\mathcal
   M_H({\mathbf v - |\nu|\delta})\times
   \underbrace{S^1_0(\RR^4)\times\dots}_{\text{$\nu_1$ times}}\times
   \underbrace{S^2_0(\RR^4)\times\dots}_{\text{$\nu_2$ times}}\times
   \cdots.\]
   Here $S^n_0(\RR^4)$ is the $n$th \emph{centered} symmetric
   power, i.e., $\big\{ (x_1, \dots, x_n)\bmod S_n \mid
     \sum x_i =
     0\big\}$.
   If we ignore the \emph{centered} condition, it is the Higgs branch
   of the quiver gauge theory in (\ref{Cs}) above. The quiver is the
   union of $Q$ and copies of the Jordan quiver, the dimension vector
   is $\mathbf v-|\nu|\delta,
   \underbrace{\delta,,\dots}_{\text{$\nu_1$ times}},
   \underbrace{2\delta,\dots}_{\text{$\nu_2$ times}}$,
   \dots.
\end{aenume}

In fact, for (\ref{Ht}), we can put the centered condition, namely we
impose that endomorphisms of $\CC^k$ in $\bM$ are trace-free. Since
Coulomb branches are unchanged if we add trivial representations to
$\bM$, we can do the same for (\ref{Cs}).

\begin{NB}
    Suppose $\mathbf v = n\delta$ for brevity. Therefore $\mathcal M_H =
    S^n(\RR^4/\Gamma)$. The open stratum is $\nu = (1^n)$. The
    transversal slice is a point, which is $S^{\mathbf v -
      |\nu|\delta}(\RR^4/\Gamma)$. The closed stratum is $\nu =
    \emptyset$. (Hence all points lie on the origin.) The transversal
    slice is $S^n(\RR^4/\Gamma) = S^{\mathbf v - |\nu|\delta}(\RR^4/\Gamma)$.

    For the stratum $\nu = (1^{n-1})$ (i.e., $n-1$ distinct points in
    $\RR^4\setminus\{0\}/\Gamma$ and one point at the origin), the
    transversal slice is $\RR^4/\Gamma = S^{\dim
      V-|\nu|\delta}(\RR^4/\Gamma)$.

    For the stratum $\nu = (n)$, i.e., a single point with
    multiplicity $n$ outside $0$, the stratum is
    $\RR^4\setminus\{0\}/\Gamma$. The transversal slice is $S^{n}_0\RR^4$.

    For the stratum $\nu = (k,l)$ ($k\neq l$), i.e., two points with
    multiplicities $k$, $l$ respectively outside $0$, the closure of
    the stratum is $\RR^4/\Gamma\times \RR^4/\Gamma$. The transversal
    slice is $S^{n-k-l}(\RR^4/\Gamma)\times S^{k}_0\RR^4\times S^{l}_0\RR^4$.

    For the stratum $\nu = (k^2)$, i.e., two distinct points with
    multiplicity $k$ outside $0$, the closure of
    the stratum is $S^2(\RR^4/\Gamma)$. The transversal
    slice is $S^{n-2k}(\RR^4/\Gamma)\times S^{k}_0\RR^4\times S^{k}_0\RR^4$.
\end{NB}

Comparing two stratifications, we find that answers to
Question~\ref{q:leaves}(1),(2) are yes, but (3) is no. The most
natural order-reversing bijection is $\nu\to\nu^t$, but the dimension
vectors do not match in (\ref{Cs}), (\ref{Ht}) and (\ref{Ct}),
(\ref{Hs}) respectively.
Even ignoring the ordering, it seems we do not have a `natural'
bijection compatible with the change of dimension vectors.
In fact, \emph{special} leaves have only $\nu =
(1^\bullet)$,\footnote{The author thanks Ben Webster for an
  explanation of this result.} hence we have an order-reversing
bijection on special leaves.

These examples presents a difficulty to generalize the description of
Coulomb branches for general type quiver gauge theories. It is a
partial compactification of the space of based maps from $\proj^1$ to
the corresponding Kac-Moody flag manifold. We need to add `defects'
which correspond to semisimple representations of the corresponding
preprojective algebra. Such a partial compactification has not been
studied before, as far as the author knows.

The answer to Question~\ref{q:complete} is yes thanks to the following:

\begin{Proposition}
    \textup{(1)} $\bmu^{-1}(0)$ is complete intersection of dimension
    $\dim\bM - \dim G$ if and only if $\mathbf v$ is either $\alpha$,
    $\delta - \alpha$ or $\delta$, where $\alpha$ is a positive root
    of the root system of finite $ADE$ type, obtained from $Q$ by
    removing the $0$-vertex.

    \textup{(2)} If $\Hyp(\bM)\tslabar G$ is good or ugly, $\mathbf v$ is
    either of the above form.

    \textup{(3)} If $\Hyp(\bM)\tslabar G$ is good, $\mathbf v$ is $\delta$.
\end{Proposition}

\begin{proof}
    (1)
    We can use the criterion \eqref{eq:8} above. If $\beta^{(k)}$ is
    an imaginary (resp.\ a real) root, $\langle\beta^{(k)},\mathbf
    C\beta^{(k)}\rangle = 0$ (resp.\ $=2$). Also positive roots are
    $n\delta + \alpha$ ($n\ge 0$), $n\delta-\alpha$ ($n > 0$),
    $n\delta$ ($n > 0$) for a positive root $\alpha$ of the underlying
    finite root system. We see that dimension vectors $\alpha$,
    $\delta-\alpha$, $\delta$ give complete intersection.

    (2) We use the same argument as in finite type case using
    \eqref{eq:10}. If $\mathbf v$ is not a positive root, we have a
    decomposition $\mathbf v = \sum \beta^{(k)}$ such that
    $\langle\beta^{(k)},\mathbf C\beta^{(l)}\rangle \ge 0$ for $k\neq
    l$ as before: We have $\langle\beta,\mathbf C\beta\rangle \ge
    0$. It is equal to $0$ (resp.\ $2$ if and only $\beta$ is an
    imaginary (resp.\ a real) root. Therefore $2\Delta(\la)\le 0$, so
    the gauge theory is not good or ugly.

    If $\mathbf v = n\delta+\alpha$ for a positive root $\alpha$ of the
    underlying finite root system and $n > 0$, we take the
    decomposition $\beta^{(1)} = n\delta$, $\beta^{(2)} =
    \alpha$. Then $\langle \beta^{(1)}, \mathbf C\beta^{(2)}\rangle =
    0$. Similarly, for $\mathbf v=n\delta-\alpha$ for $n > 1$, we take
    $\beta^{(1)} = (n-1)\delta$, $\beta^{(2)} = \delta-\alpha$ to
    deduce a contradiction. For $\mathbf v = n\delta$ with $n > 1$, consider
    $\beta^{(1)} = (n-1)\delta$, $\beta^{(2)} = \delta$.

    (3) 
    If $\mathbf v = \alpha$ or $\delta-\alpha$, it is a quiver gauge
    theory of type $ADE$. Therefore \propref{prop:good} implies that
    it is not good.
\end{proof}

For affine type $A$, we see that $\Hyp(\bM)\tslabar G$ is good if
$\mathbf v = \delta$. It is probably true in general. Assuming it, we
have $\mathcal M_H = \CC^2/\Gamma$, where $\Gamma$ is a finite
subgroup of $\SL(2)$ corresponding to $Q$. In view of
Question~\ref{q:HPHH}, it is interesting to compute $I\!H^*(\mathcal M_C)$, $H\!P_0(\mathcal M_C)$.

\subsection{Quiver gauge theories with $W\neq
  0$}\label{subsec:nonzeroW}

Let us turn to quiver gauge theories with $W\neq 0$. We take two
$Q_0$-graded vector spaces $V$, $W$. We take $G = \prod_i \GL(V_i)$
unlike the case $W=0$. When $W\neq 0$, we can deform and take
(partial) resolution of the Coulomb branch. But we set the parameter
to be $0$, and consider the most singular Coulomb branch. We denote
the dimension vector $\dim W$ by $\mathbf w$.

Assume the underlying graph is of type $ADE$. In
\cite[3(ii)]{2015arXiv150303676N}, it was conjectured that the Coulomb
branch is the moduli space of $S^1$-equivariant ${ADE}$-instantons on
$\RR^4$ where $\mathbf w$ (resp.\@ $\mathbf w - \mathbf C\mathbf v$)
corresponds to the coweight $\lambda$ (resp.\@ $\mu$) $\colon S^1\to
G_{ADE,c}$ giving the $S^1$-action on the fiber at $0$ (resp.\@
$\infty$). Here $\mathbf C$ is the Cartan matrix. This can be true
only if $\mu = \mathbf w - \mathbf C\mathbf v$ is dominant, as the
$S^1$-action on the fiber corresponds to a \emph{dominant} coweight.

After reading the preprint \cite{2015arXiv150304817B}, which appears
in arXiv shortly after \cite{2015arXiv150303676N} and then looking at
the original physics literature \cite{MR1636383}, the author
understand that the conjecture must be corrected. Namely the Coulomb
branch is the moduli space of singular $ADE$-monopoles on $\RR^3$. In
order to connect with the previous conjecture, let us recall that
singular monopoles are $S^1$-equivariant instantons on the Taub-NUT
space \cite{Kronheimer-msc}. The Taub-NUT space and $\RR^4$ are both
$\CC^2$ as a holomorphic symplectic manifold, but the Riemannian
metrics are different. It is expected that moduli spaces of
$S^1$-equivariant instantons on two spaces are isomorphic as
holomorphic symplectic manifolds when $\mu$ is dominant, but the
Taub-NUT case has more general moduli spaces, as we do not need to
assume $\mu$ is dominant. (It corresponds to the monopole
charge.)\footnote{The author thanks Sergey Cherkis for his explanation
  on instantons on (multi-)Taub-NUT spaces, a.k.a.\ bow varieties over
  years.}
The definition in \cite{BFN} produces a certain moduli space of
bundles over $\proj^1$, which conjecturally isomorphic to the moduli
space of singular monopoles.

Let us describe the stratification on $\mathcal M_C$. As in the case
with $W=0$, we need to consider the Uhlenbeck partial compactification
of the moduli space of instantons. Since we are considering
$S^1$-equivariant instantons, the bubbling only occurs at $0$, the
$S^1$-fixed point in the Taub-NUT space. 
(A bubbled instanton is defined over $\RR^4$.)
The remained $S^1$-equivariant instanton has a different dominant
coweight for the $S^1$-action on the fiber at $0$. But the coweight
for $\infty$ is unchanged. Therefore
\begin{equation*}
    \mathcal M_C(\mu,\lambda)
    = \bigsqcup_{\substack{\text{$\la'$ : dominant}
        \\\mu\le\la'\le\la}} \mathcal M_C^\circ(\mu,\lambda'),
\end{equation*}
where $\mathcal M_C^\circ(\mu,\lambda')$ denote the moduli space of
genuine $S^1$-equivariant instantons with 1) the monopole charge
$\mu$, and 2) the coweight $\lambda'$ for the $S^1$-action at $0$. For
example, $\lambda'=\lambda$ is the open stratum, and the minimum
$\lambda'\ge\mu$ is the closed stratum.

The stratum (resp.\ the transversal slice) is the Coulomb branch of
the quiver gauge theory of the same type with the dimension vectors
given by $(\mu,\lambda')$ (resp.\ $(\lambda',\lambda)$).

Let us consider the Higgs branch $\mathcal M_H = \bmu^{-1}(0)$. It is
a quiver variety of type $ADE$, and the stratification is given \cite[Prop.~6.7]{Na-quiver} as in \subsecref{sec:higgs-branch}:
\begin{equation*}
    \mathcal M_H(\mu,\lambda) = \bigsqcup_{\substack{\text{$\mu'$ : dominant}
        \\\mu\le\mu'\le\la}}
      \mathcal M_H^\circ(\mu',\lambda),
\end{equation*}
where $\mathcal M_H^\circ(\mu',\lambda) = \mathfrak
M_0^{\mathrm{reg}}(\mathbf v^{(0)},\mathbf w)$ with $\lambda = \mathbf
w$, $\mu' = \mathbf w - \mathbf C\mathbf v^{(0)}$ in the notation in
\cite{Na-quiver}.
\begin{NB}
    Since $\mathbf v^{(0)}\le \mathbf v$, $\mu = \mathbf w - \mathbf C
    \mathbf v^{(0)} \ge \mathbf w - \mathbf C \mathbf v = \mu'$.
\end{NB}%
For example, $\mu' = \lambda$ is the closed stratum and the maximum $\mu'\le\lambda$ is the open stratum.

The stratum (resp.\ the transversal slice) is the Higgs branch of the
quiver gauge theory of the same type with the dimension vectors given
by $(\mu',\lambda)$ (resp.\ $(\mu,\mu')$).

Comparing two stratifications, we find that the answer to
Question~\ref{q:leaves} is yes. The bijection is simply given by
$\lambda'=\mu'$.

There should be a similar stratification for quiver gauge theories of
affine types, where $\mathcal M_C$ is a moduli space of
$\ZZ/\ell\ZZ$-equivariant instantons on the Taub-NUT space, where
$\ell$ is the level of $\dim W$. But the author is not familiar enough
with such a moduli space, and in particular, it is not clear whether a
puzzle on the affine Dynkin diagram automorphism group raised in
\cite[3(ii)]{2015arXiv150303676N} is clarified or not. When $\mu$ is
dominant and $Q$ is of type $A$, we can consider equivariant
instantons on $\RR^4$ instead. Then we get quiver varieties of affine
type $A$. Then answers to Question~\ref{q:leaves}(1),(2) are yes,
while we have the same phenomenon as in $W=0$ case for (3).

Let us turn to Question~\ref{q:complete}. The answer is yes thanks to
\begin{Proposition}\label{prop:qWgood}
    Suppose $Q$ is finite or affine type.

    \textup{(1)} If $\Hyp(\bM)\tslabar G$ is good or ugly,
    $\bmu^{-1}(0)$ is complete intersection of dimension $\dim\bM -
    \dim G$.

    \textup{(2)} If $\Hyp(\bM)\tslabar G$ is good, $\mathbf w -
    \mathbf C\mathbf v$ is dominant.
\end{Proposition}

\begin{proof}
    Crawley-Boevey's criterion of the complete intersection property
    of $\bmu^{-1}(0)$ can be modified by the trick adding a new vertex
    $\infty$. Then instead of \eqref{eq:8}, we have
\begin{equation}\label{eq:11}
    \langle \mathbf v, 2\mathbf w - \mathbf C\mathbf v\rangle
    \ge 
    \langle \mathbf v^{(0)}, 2\mathbf w - \mathbf C\mathbf v^{(0)}\rangle
    + \sum_k (2 - \langle\beta^{(k)}, \mathbf C\beta^{(k)}\rangle)
\end{equation}
for any decomposition $\mathbf v = \mathbf v^{(0)} + \sum_k
\beta^{(k)}$ into positive vectors such that $\beta^{(k)}$ is
nonzero. ($\dim V^{(0)}$ could be zero.) (See
\cite[Th.~2.15]{Na-branching} for a similar deduction.) It is also
equivalent to the inequality holds for $\mathbf v = \mathbf v^{(0)} +
\sum_k \beta^{(k)}$ such that $\mathbf w - \mathbf C\mathbf v^{(0)}$
is a weight of the highest weight representation $V(\mathbf w)$ with
highest weight $\mathbf w$, and $\beta^{(k)}$ is a positive root.

If $Q$ is of finite type $ADE$, $\langle\beta^{(k)},\mathbf
C\beta^{(k)}\rangle = 2$ for any $k$. Therefore \eqref{eq:11} is
equivalent to
\begin{equation}\label{eq:13}
    \begin{NB}
    \langle \mathbf v, 2\mathbf w - \mathbf C\mathbf v\rangle
    \ge 
    \langle \mathbf v^{(0)}, 2\mathbf w - \mathbf C\mathbf v^{(0)}\rangle
    \end{NB}
    \langle \beta, \mathbf w - \mathbf C\mathbf v\rangle
    \ge -\frac12
    \langle \beta, \mathbf C\beta\rangle
\end{equation}
for any $\beta\le \mathbf v$ such that $\mathbf w - \mathbf C (\mathbf
v-\beta)$ is a weight of $V(\mathbf w)$.
\begin{NB}
    $\beta = \mathbf v - \mathbf v^{(0)}$
\end{NB}

\begin{NB}
Equivalently
\begin{equation*}
    \langle \sum_k \beta^{(k)},\mathbf w - \mathbf C\mathbf v\rangle
    \ge - \frac12 \langle \sum \beta^{(k)},\mathbf C\sum \beta^{(k)}\rangle
    + \sum_k (1 - \frac12 \langle\beta^{(k)}, \mathbf C\beta^{(k)}\rangle)
\end{equation*}
or
\begin{equation*}
    \langle \sum_k \beta^{(k)},\mathbf w - \mathbf C\mathbf v^{(0)}\rangle
    \ge \sum_{k < l} \langle \beta^{(k)},\mathbf C \beta^{(l)}\rangle
    + \# \{ \beta^{(k)}\}.
\end{equation*}
\end{NB}%
Taking a positive root $\beta$, we have $\langle\beta,\mathbf w-\mathbf C\mathbf v\rangle \ge -1$.
Next suppose $\beta = \beta^{(1)} + \beta^{(2)} + \cdots$ such that
$\beta^{(k)}$ is a positive root and $\beta^{(k)} + \beta^{(l)}$ is
not a root. Then $\langle \beta^{(k)},\mathbf C\beta^{(l)}\rangle \ge
0$ as in \subsecref{sec:quiv-gauge-theor}. Therefore $\langle \beta,
\mathbf C\beta\rangle \ge 2 \# \{ \beta^{(k)}\}$. Therefore if
$\langle \beta^{(k)}, \mathbf w - \mathbf C\mathbf v\rangle \ge -1$
for any $k$, we have $\langle \beta, \mathbf w - \mathbf C\mathbf
v\rangle \ge -\# \{ \beta^{(k)}\} \ge -\frac12 \langle \beta, \mathbf
C\beta\rangle$. Hence it is enough to suppose \eqref{eq:13} is true
for an arbitrary positive root $\beta$.

Suppose $Q$ is of affine type. If $\beta^{(k)}=\delta$ in
\eqref{eq:11}, we absorb it into $\mathbf v^{(0)}$. Since $\mathbf
C\delta = 0$, the first term of the right hand side of \eqref{eq:11}
increases $2\langle\delta,\mathbf w\rangle$, which is $\ge 2$. On the
other hand, the second term decreases by $2$. Therefore it is enough
to assume \eqref{eq:11} when all $\beta^{(k)}$ is a real root. Then
the same argument as above shows that it is enough to assume
$\langle\beta,\mathbf w-\mathbf C\mathbf v\rangle \ge -1$ for an
arbitrary positive real root $\beta$.

On the other hand, let us take $\la\colon\CC^\times\to T = \prod_i
T(V_i)$ as in \eqref{eq:10}. Then if $\la(t)$ acts on $t^{\la^{(k)}}$
on $V^{(k)}$,
\begin{equation}\label{eq:12}
    \begin{split}
    2\Delta(\la) &= 
    \begin{aligned}[t]
       \sum_{k<l} |\la^{(k)} - \la^{(l)}| \left( \sum_i
    -2\beta^{(k)}_i \beta^{(l)}_i
    + \sum_{i,j} a_{ij} \beta^{(k)}_i \beta^{(l)}_j
    \right) &
    \\
    + \sum_k |\la^{(k)}| \sum_i \beta^{(k)}_i \dim W_i &
    \end{aligned}
\\
    & = -\sum_{k<l} |\la^{(k)} - \la^{(l)}| \langle \beta^{(k)}, \mathbf C
    \beta^{(l)}\rangle
    + \sum_k |\la^{(k)}|\langle\beta^{(k)},\mathbf w\rangle
    \end{split}
\end{equation}
with $\beta^{(k)} = \dim V^{(k)}$.
\begin{NB}
We take $\beta^{(1)} = \alpha_i$, $\beta^{(2)} = \mathbf v -
\alpha_i$. Then
\begin{equation*}
    2\Delta(\la) = - |\la^{(1)}-\la^{(2)}|
    \langle\alpha_i,\mathbf C(\mathbf v-\alpha_i)\rangle
    + |\la^{(1)}| \langle\alpha_i,\mathbf w\rangle
    + |\la^{(2)}| \langle\mathbf v - \alpha_i,\mathbf w\rangle.
\end{equation*}
More generally, suppose we have only two summands $\beta^{(1)}$,
$\beta^{(2)} = \mathbf v - \beta^{(1)}$. Then
\begin{equation*}
    2\Delta(\la) = - |\la^{(1)}-\la^{(2)}|
    \langle\beta^{(1)},\mathbf C(\mathbf v-\beta^{(1)})\rangle
    + |\la^{(1)}| \langle\beta^{(1)},\mathbf w\rangle
    + |\la^{(2)}| \langle\mathbf v - \beta^{(1)},\mathbf w\rangle.
\end{equation*}

Suppose $\la^{(1)} \ge \la^{(2)}\ge 0$:
\begin{equation*}
    \begin{split}
    2\Delta(\la) &= - (\la^{(1)}-\la^{(2)})
    \langle\beta^{(1)},\mathbf C(\mathbf v-\beta^{(1)})\rangle
    + \la^{(1)} \langle\beta^{(1)},\mathbf w\rangle
    + \la^{(2)} \langle\mathbf v - \beta^{(1)},\mathbf w\rangle
\\
   &=  \la^{(1)} \langle\beta^{(1)},\mathbf w 
   - \mathbf C(\mathbf v-\beta^{(1)})\rangle
   + \la^{(2)} \langle\mathbf w + \mathbf C\beta^{(1)},
   \mathbf v-\beta^{(1)}\rangle
\\
   &=  \la^{(1)} \langle\beta^{(1)},\mathbf w 
   - \mathbf C(\mathbf v-\beta^{(1)})\rangle
   + \la^{(2)} \langle\mathbf w - \mathbf C(\mathbf v - \beta^{(2)}),
   \beta^{(2)}\rangle.
    \end{split}
\end{equation*}
Therefore $\langle \beta^{(1)},\mathbf w 
   - \mathbf C(\mathbf v-\beta^{(1)})\rangle\ge 1$, i.e., 
\(
   \langle \beta^{(1)},\mathbf w 
   - \mathbf C\mathbf v\rangle\ge 
   1 - \langle \beta^{(1)}, \mathbf C\beta^{(1)}\rangle.
\)

Suppose $\la^{(1)} \ge 0\ge \la^{(2)}$:
\begin{equation*}
    \begin{split}
    2\Delta(\la) &= - (\la^{(1)}-\la^{(2)})
    \langle\beta^{(1)},\mathbf C(\mathbf v-\beta^{(1)})\rangle
    + \la^{(1)} \langle\beta^{(1)},\mathbf w\rangle
    - \la^{(2)} \langle\mathbf v - \beta^{(1)},\mathbf w\rangle
\\
   &=  \la^{(1)} \langle\beta^{(1)},\mathbf w 
   - \mathbf C(\mathbf v-\beta^{(1)})\rangle
   - \la^{(2)} \langle\mathbf w - \mathbf C\beta^{(1)},
   \mathbf v-\beta^{(1)}\rangle
\\
   &=  \la^{(1)} \langle\beta^{(1)},\mathbf w 
   - \mathbf C(\mathbf v-\beta^{(1)})\rangle
   - \la^{(2)} \langle\mathbf w - \mathbf C(\mathbf v - \beta^{(2)}),
   \beta^{(2)}\rangle.
    \end{split}
\end{equation*}
Suppose $0\ge \la^{(1)} \ge \la^{(2)}$:
\begin{equation*}
    \begin{split}
    2\Delta(\la) &= - (\la^{(1)}-\la^{(2)})
    \langle\beta^{(1)},\mathbf C(\mathbf v-\beta^{(1)})\rangle
    - \la^{(1)} \langle\beta^{(1)},\mathbf w\rangle
    - \la^{(2)} \langle\mathbf v - \beta^{(1)},\mathbf w\rangle
\\
   &=  -\la^{(1)} \langle\beta^{(1)},\mathbf w 
   + \mathbf C(\mathbf v-\beta^{(1)})\rangle
   - \la^{(2)} \langle\mathbf w - \mathbf C\beta^{(1)},
   \mathbf v-\beta^{(1)}\rangle.
    \end{split}
\end{equation*}
\end{NB}%
We take $\beta^{(1)} = \beta$ a positive real root, and $\beta^{(2)} =
\mathbf v - \beta$. Furthermore assume $\la^{(1)}\ge \la^{(2)}\ge
0$. Then
\begin{equation*}
    2\Delta(\la) = 
    \la^{(1)} \langle\beta,\mathbf w 
    - \mathbf C(\mathbf v-\beta)\rangle
    + \la^{(2)} \langle\mathbf w + \mathbf C\beta,
    \mathbf v-\beta\rangle.
\end{equation*}
For good or ugly cases, as $2\Delta(\la)\ge 1$ for any $\la$, we have
$\langle\beta,\mathbf w - \mathbf C(\mathbf v-\beta)\rangle\ge 1$,
i.e., $\langle\beta,\mathbf w - \mathbf C\mathbf v\rangle\ge -1$. For
good cases, $\langle\beta,\mathbf w - \mathbf C\mathbf v\rangle\ge 0$.
\end{proof}

The converses of (1),(2) are probably true.

Suppose $\Hyp(\bM)\tslabar G$ is good. By \propref{prop:qWgood}(2),
$\mathbf w-\mathbf C\mathbf v$ is dominant, hence $\mathcal M_C$ is
expected to be a slice in the affine Grassmannian when $Q$ is finite
type, as we explained in the beginning of this subsection. When $Q$ is
affine type, we still need to solve a puzzle in
\cite[3(ii)]{2015arXiv150303676N}, but is the Uhlenbeck partial
compactification of an instanton moduli space on $\RR^4/(\ZZ/\ell\ZZ)$
as a first approximation. Then $I\!H^*(\mathcal M_C)$ is a weight
space of a finite dimensional irreducible representation of the Lie
algebra $\g$ corresponding to $Q$ by geometric Satake
correspondence. This is so when $Q$ is of finite type. If $Q$ is
affine type, this is the statement of a conjecture in
\cite{braverman-2007}, geometric Satake correspondence for the affine
Lie algebra $\g_\aff$.

On the other hand, $\mathcal M_H$ is a quiver variety. In particular,
$\mathcal M_H$ has a symplectic resolution $\tilde{\mathcal
  M}_H\to\mathcal M_H$. It is conjectured that $H\!P_0(\mathcal M_H)
\cong H^{\dim \tilde{\mathcal M}_H}(\tilde{\mathcal M}_H)$ in
\cite{MR3217666}. The right hand side is a weight space of a finite
dimensional or integrable irreducible representation of $\g$ or
$\g_\aff$ by \cite{Na-quiver}. Therefore, modulo a conjecture in
\cite{MR3217666}, we have $H\! P_0(\mathcal M_H) \cong I\!H^*(\mathcal
M_C)$, the first isomorphism in Question~\ref{q:HPHH}. We also expect
$H\!P_0(\mathcal M_C)\cong I\!H^*(\mathcal M_H)$, where the right hand
side is the multiplicity of the finite dimensional (resp.\ integrable)
irreducible representation $L(\mathbf w - \mathbf C\mathbf v)$ of $\g$
(resp.\ $\g_\aff$) in the (resp.\ affine) Yangian $Y(\g)$ (resp.\
$Y(\g_\aff)$) \cite[\S15]{Na-qaff}. It is interesting to study
$HP_0(\mathcal M_C)$.

\subsection{$\SU(2)$ gauge theories with fundamental matters}\label{ex:SL2}
Consider $(G,\bN) = (\SL(2), (\CC^2)^{\oplus N})$ with $\bM =
\bN\oplus\bN^*$ for $N=0,1,2,\dots$.
In this case, the Coulomb branch is a complex surface
\begin{equation}\label{eq:2}
    y^2 = x^2 z - z^{N-1} \quad \text{if $N\ge 1$}, \quad
    y^2 = x^2 z + x \quad \text{if $N=0$}.
\end{equation}
See \cite{MR1490862}. This is of cotangent type, and the definition in
\cite{BFN} reproduces the above at least for $N\neq 1,2,3$. The degrees
for the $\CC^\times$-action are
\(
   \deg x = N-2,
\)
\(
   \deg y = N-1,
\)
\(
   \deg z = 2.
\)

\begin{NB}
    If $N\ge 3$ (good), we have
    \(
        \dim I\!H^*(\mathcal M_C) = 1.
    \)
\end{NB}

Let us study the Higgs branch $\bM\tslash G$. We use the standard
notation for quiver varieties: $i\colon \CC^N\to \CC^2$, $j\colon
\CC^2\to \CC^N$ with the $\SL(2)$-action by $g\cdot (i,j) = (gi,
jg^{-1})$. The moment map $\bmu(i,j)$ is the trace-free part of
$ij$. If $N=0$, we have $\bM\tslash G = \{ 0\}$ by a trivial
reason. If $N=1$, it is not trivial, but not difficult to check the
following:
\begin{Lemma}
    Suppose $N=1$.

    Then $\bmu=0$ implies either $i=0$ or $j=0$. In particular, the
    only closed $\SL(2)$-orbit in $\bmu^{-1}(0)$ is just
    $0$. Therefore $\bM\tslash G = \{ 0\}$.
\end{Lemma}

Therefore $\bmu^{-1}(0)$ is not complete intersection of $\dim = \dim
\bM - \dim G = 1$. The degree of $x$ is $-1$, hence $\mathcal M_C$ is
not conical.

On the other hand, since the Higgs branch has only single point,
$\mathcal M_C(G,\bN)$ should have only one stratum, i.e., it is a
nonsingular symplectic manifold. It is not difficult to check that
\eqref{eq:2} is indeed so when $N=0$, $1$. Therefore
Question~\ref{q:smooth} is affirmative.

Next consider the case $N\ge 2$. 

\begin{Proposition}
    Assume $N\ge 2$.

    \textup{(1)} $\bmu^{-1}(0)$ is a complete intersection in $\bM$ of
    $\dim = \dim \bM - 3 = 4N-3$. It is irreducible if $N\ge 3$ and has two
    irreducible components if $N=2$.

    \textup{(2)} $\bM\tslash G = \bmu^{-1}(0)\dslash G$ has
    singularity only at $0$. $\bM\tslash G\setminus\{0\}$ is
    irreducible if $N > 2$ and has two irreducible components if $N=2$.
\end{Proposition}

Since the gauge theory $\Hyp(\bM)\tslabar G$ is good or ugly if and
only if $N\ge 3$, so the answer to Question~\ref{q:complete} is yes,
but $N=2$ is a counter-example to its converse.

\begin{proof}
    (1) From $\bmu(i,j) = 0$, we have $ij=\zeta\id$ for some
    $\zeta\in\CC$. Suppose $\zeta\neq 0$. A standard argument shows
    that the stabilizer is trivial. Therefore the differential of
    $\bmu$ is surjective, hence $\bmu^{-1}(0)$ is a complete
    intersection.
    If we take the quotient of $ij=\zeta\id$ by $\GL(2)$, it is a
    quiver variety $\mathfrak M_\zeta$ and is smooth and irreducible
    and forms a smooth family over $\zeta\neq 0$. (In fact, it is a
    semisimple coadjoint orbit in $\gl(N)$ with eigenvalues $\la$ with
    multiplicity $2$ and $0$ with multiplicity $N-2$.)

    If $i$ is surjective, we can form the (GIT) quotient even across
    $\zeta=0$ and get a smooth family over $\zeta\in\CC$. The same is
    true if $j$ is injective. Therefore on the open subset either of
    $i$ or $j$ is rank $2$, $\bmu^{-1}(0)$ is a smooth irreducible
    \begin{NB}
        (as connected)
    \end{NB}%
    variety of $\dim = 4N - 3$.

    Therefore consider the case when both $i$, $j$ have rank $\le
    1$. This can happen only when $\zeta = 0$. Let us write
    \begin{equation*}
        i =
    \begin{pmatrix}
        i_{11} & i_{12} & \cdots & i_{1N} \\
        i_{21} & i_{22} & \cdots & i_{2N}
    \end{pmatrix}, \quad
    j = 
    \begin{pmatrix}
        j_{11} & j_{12} \\
        \vdots & \vdots \\
        j_{N1} & i_{N2}
    \end{pmatrix}.
    \end{equation*}
    Since $i$ is rank $1$, the upper row and the lower row are the
    same up to constant multiple. The same is true for columns of
    $j$. Then the equation $ij=0$ is a single scalar
    equation. Therefore it forms an irreducible variety of dimension
    $2N+1$. Since $2N+1 \le 4N-3$ and the equality holds if and only
    if $N=2$, the assertion follows.
    \begin{NB}
        Since it is defined by $3$ equations, we cannot have an
        irreducible component of $\dim < \dim M - 3$.
    \end{NB}%

    (2) We suppose $(i,j)\in\bmu^{-1}(0)$ corresponds to a singular
    point. We assume that it has a closed $\SL(2)$-orbit. By the above
    argument, we may suppose both $i$, $j$ have rank $1$. If $\Ima
    i\cap \Ker j = \{ 0\}$, the stabilizer is trivial.
    \begin{NB}
        It is also true that it has a closed orbit.
    \end{NB}%
    Therefore it gives a smooth point in $\bM\tslash G$. If $\Ima i =
    \Ker j$, we can find a one-parameter subgroup $\la\colon
    \CC^\times\to \SL(2)$ such that $\la\cdot (i,j)\to
    (0,0)$. Therefore it cannot have a closed orbit.
    
    There are indeed points $\Ima i \cap \Ker j = \{ 0\}$ in
    $\bmu^{-1}(0)$. They form a smooth variety of dimension
    $2N-2$. When $N=2$, it gives the second irreducible component.
\end{proof}

This example shows that a stratum $(\bM\tslash G)_{(\widehat G)}$ may
not be connected in general. In fact, we have $\widehat G = \{e\}$ in
the above example, and $(\bM\tslash G)_{(\widehat G)}$ is an open
stratum and has two connected components.
Therefore $\bM\tslash G$ actually has three strata.
It is interesting to note that \eqref{eq:2} with $N=2$ also has three
strata, a smooth locus and two singular points $x=\pm 1$,
$y=z=0$. Thus answers to Question~\ref{q:leaves}(1) and the bijection
part of (3) are yes even in this case.

(2) and the second half of (3) are \emph{not} clear as stated: The
transversal slice $T^\perp/\widehat G$, for $\mathcal M_H$ at points
in either of two components of $(\bM\tslash G)_{\{e\}}$, is
$\{ 0\}/\{e\}$. It is the Higgs branch of the trivial gauge theory
$\Hyp(\{0\})\tslabar \{e\}$. The corresponding Coulomb branch is just
a single point, and strata for two singular points may be identified
with this Coulomb branch.
However it is not clear (at least to the author) whether we can
naturally view each component of $(\bM\tslash G)_{\{e\}}$ is the Higgs
branch of a gauge theory $\Hyp(\bM')\tslabar G'$ for some
$(G',\bM')$. Say, are $(G',\bM')$ different for two components ?
Similarly it is not clear whether the transversal slices to singular
points $x=\pm 1$, $y=z=0$ in $\mathcal M_C$ (both $A_1$ type) can be
naturally identified with the Coulomb branch of a gauge theory
$\Hyp(\bM')\tslabar G'$.
It is desirable to understand this phenomenon better.

For Question~\ref{q:HPHH}, $H\!P_0(\mathcal M_C)$, $I\!H^*(\mathcal
M_C)$ are easy to compute as $\mathcal M_C$ is of type $D_N$
singularity. We do not know $I\!H^*(\mathcal M_H)$, $H\!P_0(\mathcal
M_H)$.

\subsection*{Acknowledgements}

The author has been taught many things from various people at various
places in the world while he was explaining
\cite{2015arXiv150303676N,BFN}. Besides those who are already
mentioned, he thanks
Mathew Bullimore,
Tudor Dimofte,
Davide Gaiotto,
Amihay Hanany,
Joel Kamnitzer,
Nicholas Proudfoot,
Bal\'as Szenrd\"oi,
and
Yuuya Takayama.
He also thanks Alexander Braverman
and
Michael Finkelberg
for the collaboration \cite{BFN}.

This research is supported by JSPS Kakenhi Grant Numbers 
23224002, 
23340005, 
24224001, 
25220701. 

\bibliographystyle{myamsalpha}
\bibliography{nakajima,mybib,coulomb}

\end{document}